\documentclass[runningheads,a4paper]{llncs}
\usepackage{amssymb, amsmath, graphicx}
\usepackage[classfont=caps, langfont=caps, small]{complexity}
\usepackage{changepage}   % for the adjustwidth environment

\newtheorem{theorem}{Theorem}
\newtheorem{lemma}[theorem]{Lemma}
\newtheorem{corollary}[theorem]{Corollary}
\newtheorem{definition}[theorem]{Definition}
\newtheorem{question}[theorem]{Question}

\def\picture#1#2#3{%
\begin{figure}[!ht]
  \centering
  \includegraphics[width=#3 \textwidth]{#1.eps}
  \caption{#2}
  \label{f:#1}
\end{figure}}

\def\dir{\overrightarrow}

\newcommand{\keywords}[1]{\par\addvspace\baselineskip
\noindent\keywordname\enspace\ignorespaces#1}

\def\gadget#1{\smallskip\noindent{\bf #1 gadget:}}
\def\prereq#1{\smallskip\noindent{\bf Prerequisites of #1:}}
\def\scenario#1{\smallskip\noindent{\bf Scenario of #1:}}

\def\calF{{\cal F}}
\def\calT{{\cal T}}
\def\calA{{\cal A}}
\newclass{\GuardDir}{GuardDir}
\newclass{\GuardUndir}{GuardUndir}
\newclass{\Guard}{Guard}
\newclass{\GuardUndirPSP}{GuardUndirPSP}
\newclass{\GuardPSP}{GuardPSP}
\newclass{\GuardRev}{GuardRev}
\newclass{\FormSAT}{FormSAT}
\newclass{\EXPTIME}{EXPTIME}

\def\COM{\mathit{COM}}

\def\VAR{\mathit{VAR}}

\def\ROBG{\mathit{ROBG}}
\def\ROBS{\mathit{ROBS}}
\def\COPG{\mathit{COPG}}
\def\COPS{\mathit{COPS}}
\def\BLOCK{\mathit{BLOCK}}

\def\FORCE{\mathit{FORCE}}
\def\var#1{\mathit{#1}}

\def\G{\mathcal{G}}
\def\F{\mathcal{F}}
\def\R{\mathbb{R}}
\def\N{\mathbb{N}}
\def\logred{\preceq_{\log}}

\begin{document}

\mainmatter

\title{The guarding game is \E-complete}
\author{R. \v S\'amal\inst{1}
\thanks{
  Partially supported by Karel Jane\v cek Science \& Research Endowment (NFKJ) grant 201201,
  by grant LL1201 ERC CZ of the Czech Ministry of Education, Youth and Sports and
  by grant GA \v{C}R P202-12-G061.}
\and T. Valla\inst{2}
\thanks{
Supported by the Centre of Excellence -- Inst.\ for Theor.\ Comp.\ Sci.
(project P202/12/G061 of GA~\v{C}R),
and by the GAUK Project 66010 of Charles University in Prague.}
}

\institute{
Computer Science Institute (CSI) of Charles University,\\
Malostransk\'e n\'am.~2/25, 118~00, Prague, Czech Republic \\
\email{samal@iuuk.mff.cuni.cz}
\medskip
\and
Faculty of Information Technology,\\
Czech Technical University,
Th\'akurova 9, 160~00, Prague~6, Czech Republic\\
\email{tomas.valla@fit.cvut.cz}
}

%\date{22. 4. 2013}

\maketitle

\begin{abstract}
The guarding game is a game in which several cops try to guard a region
in a (directed or undirected) graph against Robber.
Robber and the cops are placed on the vertices of the graph; they take turns
in moving to adjacent vertices (or staying), cops inside the guarded region, Robber
on the remaining vertices (the robber-region).
The goal of Robber is to enter the guarded region at a vertex with no cop on it.
The problem is to determine whether for a given graph and given number of cops
the cops are able to prevent Robber from entering the guarded region.
Fomin et al.\ [Fomin, Golovach, Hall, Mihal\'ak, Vicari, Widmayer: How to Guard a Graph?
Algorithmica 61(4), 839--856 (2011)]
proved that the problem is \NP-complete when the robber-region is restricted to a tree.
Further they prove that is it \PSPACE-complete when the robber-region is restricted
to a directed acyclic graph, and they ask about the problem complexity for arbitrary graphs.
In this paper we prove that the problem is \E-complete for arbitrary directed graphs.
\footnote[3]{This paper is based on an extended abstract which has appeared
in Proceedings of IWOCA 2011, Springer-Verlag LNCS \cite{SSV}.}
\keywords{pursuit game, cops and robber game, graph guarding game, computational complexity,
E-completeness, games on propositional formulas}
\end{abstract}

\section{Introduction and motivation}

The \emph{guarding game} $(G,V_C,c)$, introduced by Fomin et al.~\cite{HGG}, is
played on a graph $G=(V,E)$ (or directed graph $\dir G=(V,E)$) by two players,
the \emph{cop-player} and the \emph{robber-player},
each having his pawns ($c$ cops and one Robber, respectively) on~$V$.
There is a protected region (also called cop-region) $V_C\subset V$.
The remaining region $V\setminus V_C$ is called robber-region and denoted $V_R$.
Robber aims to enter $V_C$ by a move to a vertex of $V_C$ with no cop on it.
The cops try to prevent this.
The game is played in alternating turns.
In the first turn the robber-player places Robber on some vertex of $V_R$.
In the second turn the cop-player places his $c$ cops on vertices of $V_C$
(more cops can share one vertex).
In each subsequent turn the respective player can move each of his pawns to
a neighbouring vertex of the pawn's position (or leave it where it is).
However, the cops can move only inside $V_C$ and Robber can move only
on vertices with no cops.
At any time of the game both players know the positions of all pawns.
The robber-player wins if he is able to move Robber to some vertex of $V_C$
in a finite number of steps.
The cop-player wins if the cop-player can prevent the robber-player from
placing Robber on a vertex in $V_C$ indefinitely.
Note that $2|V|^{c+1}$ is the upper bound on the number of all possible positions of Robber and all cops,
so after that many turns the position has to repeat. Thus, if Robber can win,
he can win in less than $2|V|^{c+1}$ turns.
Consequently, we may define Robber to lose if he does not win in $2|V|^{c+1}$ turns.
Without loss of generality, we may assume that Robber never passes his move.

For a given graph $G$ and guarded region $V_C$, the task is to find
the minimum number $c$ such that cop-player wins. Note that this problem
is polynomially equivalent with the problem
of determining the outcome of the game for a fixed number $c$ of cops.

The guarding game is a member of a big class called pursuit-evasion games,
see, e.g., Alspach~\cite{Als} or Bonato and Nowakowski~\cite{BoNo} for introduction and survey.
Pursuit-evasion games is a family of problems in which one group attempts to
track down members of another group in some environment.
Probably the best known variant of these games is called the Cops-and-Robber game.
In this game a graph $G$ is given, together with the number $c$ of cops
and one robber. The cops are placed at the vertices in the first turn, then the robber is placed in the second turn.
In alternating turns, the cops may move to a neighbouring vertex, and the robber may move to a neighbouring vertex.
The question is whether the cops have a strategy to capture the robber, which means at some point
enter a vertex with the robber.

The Cops-and-Robber game was first defined for one cop by Winkler and Nowakowski~\cite{WiNo} and by Quilliot~\cite{Qui}.
Aigner and Fromme~\cite{AigFro} initiated the study of the problem with several cops.
The minimum number of cops required to capture the robber in a given graph is called the cop number of the graph.
The guarding game is thus a natural variant of the Cops-and-Robber game.
The complexity of the decision problem related to the Cops-and-Robber game
was studied by Goldstein and Reingold~\cite{GoRe}. They have shown that
if the number of cops is not fixed and if either the graph is directed or
initial positions are given, then the problem is \E-complete.
Also, recently Mamino~\cite{Mam} has shown that the Cops-and-Robber decision problem
without any further restriction is \PSPACE-hard.
Very recently, Kinnersley~\cite{Kinn} published a proof that the computational complexity
of Cops-and-Robber decision problem without further restrictions is \E-complete.
Unfortunately, the previous papers \cite{GoRe,Kinn} use a misleading notion, using \EXPTIME\ instead of \E.

Another interesting variant is the ``fast robber'' game,
which is studied in Fomin et al.~\cite{FastR}.
In this version, the robber may move $s$-times faster than the cops.
The authors prove that the problem of computing the minimum number of cops
such that the fast Cops-and-Robber game is cop-win, is \NP-complete
and the parameterised version (where the parameter is the number of cops) is \W[2]-complete.
See the annotated bibliography~\cite{FoThi} for reference on further topics.

A different well-studied problem, the \emph{Eternal Domination} problem
(also known as \emph{Eternal Security}) is strongly related to the guarding game.
The objective in the Eternal Domination, introduced by Burger et al.~\cite{ED2},
is to place the minimum number of guards
on the vertices of a graph $G$ such that the guards can protect the vertices of $G$
from an infinite sequence of attacks. In response to an attack of an unguarded vertex $v$,
at least one guard must move to $v$ and the other guards can either stay put,
or move to adjacent vertices.
The minimum number of guards needed to protect the vertices of $G$ is denoted by $\gamma_\infty(G)$.

The Eternal Domination problem is in fact a special case of the guarding game.
This can be seen as follows. Let $G$ be a graph on $n$ vertices and we construct a graph $H$
from $G$ by adding a clique $K_n$ on $n$ vertices and connecting the clique and $G$
by $n$ edges which form a perfect matching.
The cop-region is $V(G)$ and the robber-region is $V(K_n)$. Now $G$ has an eternal
dominating set of size $k$ if and only if $k$ cops can guard $V(G)$.

Several variants of the Eternal Domination are studied and various results are known.
For example in \cite{ED1}, the variant when only one guard may move in a turn is studied,
and it is shown that all graphs not having $K_4$ as a minor satisfy $\gamma^1_\infty(G)=\theta(G)$
where $\theta(G)$ denotes the clique covering number of $G$ and $\gamma^1_\infty(G)$ denotes
the minimum number of guards needed to protect $G$ when only one guard may move in a turn
(so called \emph{eternal $1$-security}).
Also, for the eternal 1-security, $\alpha(G)\le\gamma^1_\infty(G)\le \theta(G)$ holds for every graph $G$
where $\alpha(G)$ is the independence number,
and if $\alpha(G)=2$ then $\gamma^1_\infty(G)\le 3$ \cite{ED3}.
Klostermeyer and MacGillivray~\cite{ED6} proved that $\gamma^1_\infty(G) \le \alpha(G)(\alpha(G)+1)/2$.
In \cite{ED4} it is shown that
there are graphs $G$ for which $\gamma^1_\infty(G)\ge{\alpha(G)+1\choose 2}$.
In \cite{ED3} the following conjecture is posed: if $\alpha(G)=\gamma^1_\infty(G)$ for a graph $G$,
then the $\alpha(G)$ is equal to the chromatic number of the complement of the graph.
This conjecture was proved in \cite{ED7} for graphs with $\alpha(G)=2$.

Various complexity-related question were also studied. For example, in \cite{ED5}
it is shown that the decision problem of the eternal 1-security is \co-$\NP^\NP$-hard
(in fact, there is a completeness proof in \cite{ED5} but later an error in the proof was found).
Further graph-theory and algorithmic issues are discussed in \cite{ED8}:
For the eternal 1-security, the relationship between the number of guards and the independence and
clique covering numbers of the graph is studied, results concerning which
triples of these parameters can be attained by some graph are given, and the
exact value of the number of guards for graphs in certain classes is determined.
For the Eternal Domination, a linear algorithm to determine the minimum number
of guards necessary to defend a tree is given in \cite{ED8}.

\smallskip

In our paper we focus on the complexity issues of the following decision problem:
Given the guarding game $\G=(G,V_C,c)$, who has the winning strategy?

Let us define the computational problem precisely.
In measuring the computational complexity, our model is the
deterministic one-tape Turing Machine.

\begin{definition}
The \emph{guarding decision problem} (\Guard) is,
given a guarding game $(\dir G,V_C,c)$ where $\dir G$ is a directed graph, to decide
whether it is a cop-win game or a robber-win game.
%Analogously, we define the \emph{undirected guarding decision problem} (\GuardUndir) with the difference that
%the underlying graph $G$ is undirected.
%Similarly we define the \emph{guarding decision problem} (\Guard) when there are no restrictions on the
%underlying graph.
More precisely,
we specify a (binary) encoding of the graph $\dir G$, subset $V_C$ and starting positions;
the language \Guard\ then consists of all words
that encode a configuration winning for the robber-player.
\end{definition}

We may also define the \emph{guarding problem}, which is,
given a directed or undirected graph $G$ and a cop-region $V_C\subseteq V(G)$,
to compute the minimum number $c$ such that the $(G,V_C,c)$ is a cop-win.
It is easy to see, that the guarding problem can be in polynomial time reduced
to \Guard\ by trying all possible values of $c$.
\Guard\ was introduced and studied by Fomin et al.~\cite{HGG}.
The computational complexity of \Guard\ depends heavily on the chosen restrictions on the graph $G$.
In particular, in~\cite{HGG} the authors show that
if Robber's region is only a
path, then the problem can be solved in polynomial time, and when Robber moves in a tree
(or even in a star), then the problem is \NP-complete. Furthermore, if
Robber is moving in a directed acyclic graph, the problem becomes \PSPACE-complete.
Later Fomin, Golovach and Lokshtanov~\cite{Intruder} studied the \emph{reverse guarding game}
with the same rules as in the guarding game, except that the cop-player plays first.
They proved in \cite{Intruder} that the related decision problem is \PSPACE-hard on undirected graphs.
Nagamochi \cite{Naga} has also shown that that the problem is \NP-complete
even if $V_R$ induces a 3-star and that the problem is polynomially solvable if
$V_R$ induces a cycle.
Also, Thirumala Reddy, Sai Krishna and Pandu Rangan proved \cite{TSP}
that if the robber-region is an arbitrary undirected graph, then the decision
problem is \PSPACE-hard.

Fomin et al.~\cite{HGG} asked the following question.

\begin{question}\label{q:fomin}
Is \Guard\ \PSPACE-complete?
\end{question}

Let us consider the class $\E=\DTIME(2^{O(n)})$ of languages recognisable by a deterministic
Turing machine in time $2^{O(n)}$, where $n$ is the input size.
In pursuit of Question~\ref{q:fomin} we prove the following result.

\begin{theorem}\label{t:main}
\Guard\ is \E-complete under log-space reductions.
Furthermore, $\E\logred\Guard$ via length order $n(\log n)^2$.
\end{theorem}

We defer the precise technical definitions of log-space reductions and length
order to the end of Section~\ref{s:mainproof}.
Note that the log-space reduction via certain length order
is a more subtle definition of reduction than the general log-space reduction.
The motivation is as follows.
The reduction $\mathcal{L}\logred B$ via length order $\ell(n)$ means
that given some instance $A\in L\in \mathcal{L}$, we obtain an
instance $B$ equivalent to $A$, which size has increased to $\ell(|A|)$ (where $|A|$ is measured in bits).
Therefore, informally, given a solver $S$ for $L$ working in time $2^{O(n)}$,
this proves that the time needed to solve $B$ is at least $2^{O(\ell^{-1}(n))}$
where $\ell^{-1}$ is the inverse function to $\ell$.

Immediately, we get the following corollary.

\begin{corollary}\label{c:main}
The guarding problem is \E-complete under log-space reductions.
\end{corollary}

Using the time-hierarchy theorem, we get the following consequence (for the proof see Section~\ref{s:mainproof}).

\begin{corollary}\label{c:maintime}
There is a constant $c>1$ so that if a deterministic one-tape Turing machine
decides \Guard\ within time $t(n)$, then $t(n) > c^{n/(\log n)^2}$
for infinitely many~$n$.
\end{corollary}

Let us explain here the relevance of Theorem~\ref{t:main} to Question~\ref{q:fomin}.
Very little is known about how the class \E\ is related to \PSPACE.
It is only known \cite{Book} that $\E\neq \PSPACE$.
The following corollary shows that positive answer to Question~\ref{q:fomin}
would give a relation between these two complexity classes.
This gives unexpected and strong incentive to find positive answer to Question~\ref{q:fomin}.
(On the other hand, to the skeptics among us, it may also indicate that negative answer is more likely.)

\begin{corollary}
If the answer to Question~\ref{q:fomin} is yes, then $\E \subseteq \PSPACE$.
\end{corollary}

\begin{proof}
Suppose the guarding problem is \PSPACE-complete.
Let $L\in \E$.
Then (by Theorem~\ref{t:main}) an instance of $L$ can be reduced by a log-space reduction to an instance of the
guarding game, which we suppose to be in \PSPACE. Consequently, $L \in \PSPACE$.
\qed\end{proof}

Let us mention here, that the distance between \PSPACE\ and \E\ is rather small
in the realm of guarding games: The game, that we will construct to show \E-hardness,
has the property that after removing a single edge, the robber-region becomes
a directed acyclic graph, and for such, the decision problem is in \PSPACE\ by
a result of Fomin et al.~\cite{HGG}.

In a sense, we show analogous result for the guarding game as Goldstein and Reingold~\cite{GoRe}
have shown for the original Cops-and-Robber game.

\section{The \E-completeness proof}

\subsection{An algorithm in \E}

In order to prove \E-completeness of \Guard,
we first note that the problem is in \E.

\begin{definition}
We define the \emph{guarding game with prescribed starting positions} $\G=(G,V_C,c,S,r)$,
where $S:\{1,\dots,c\}\to V_C$ is the initial placement of cops and $r\in V_R$ is
the initial placement Robber, and we consider the guarding game $(G,V_C,c)$ after the two initial
turns where the cops were placed as described by $S$ and Robber was placed to $r$.
The \emph{guarding decision problem with prescribed starting positions} (\GuardPSP) is,
given a guarding game with prescribed starting positions $(\dir G,V_C,c,S,r)$ where $\dir G$ is a directed graph,
to decide whether it is a cop-win game or a robber-win game.
%The \emph{undirected guarding decision problem with prescribed starting positions}  (\GuardUndirPSP) is defined analogously.
\end{definition}

\begin{lemma}\label{l:etime}
The guarding decision problem is in \E.
\end{lemma}

\begin{proof}
We need to show that there is an algorithm deciding the outcome of a given guarding game $\G=(G,V_C,c)$ in $2^{O(n)}$ time,
where $n$ is the size of the input~$\G$ in some encoding.
Let us first consider the normal (i.e., not reverse) version of the problem.
Consider the directed graph $H$ of all configurations of the game~$\G$
where edges mark the legal moves of $\G$.

%-- the vertices of $H$ are all possible legal positions
%of all cops and the robber, together with the information whose turn it is. There is also a starting vertex~$s$
%representing the empty board and the vertices $r_1,\dots,r_{|V_R|}$
%representing every possible initial placement of the robber with still no cops placed.
%More precisely,
%$$
%V(H)=\left\{(C,r,t) \mid  C\colon\{1,\dots,c\}\to V_C,\, r\in V(G),\, t\in\{0,1\}\right\}
 %\cup \{s,r_1,\dots,r_{|V_R|}\}.
%$$
%Here $t=0$ denotes it is the robber's turn, $t=1$ denotes it is the cop's turn, $C$~is the position
%of cops and $r$ is the position of the robber.
%There are edges from~$s$ to every vertex $r_i$ and for every $r_i$ there are
%edges to every possible initial subsequent placement of cops.
%The edge $(k_1,k_2)$ belongs to $E(H)$ if and only if $k_1$ is cop turn and $k_2$ is robber turn (or vice versa)
%and the pawns of $k_1$ can be legally moved into pawn positions of $k_2$.

We use the standard backwards-labelling algorithm: first we label all final states of the game,
then all states that can be finished by a single move, etc. In this way we can decide the outcome
of every configuration in time polynomial in the size of the graph $H$.

%
%Let us denote the robber-winning configurations by $W_R$.
%\begin{enumerate}
%\item Construct the graph $H$.
%\item Initially set $W_R$ to be all vertices that are a win for the robber-player, i.e., positions where the
%robber stands on some $v\in V_C$ and there is no cop on $v$.
%\item Add to $W_R$ all vertices $v$ where it is the robber's turn and there is an edge $(v,w)\in E(H)$ and $w\in W_R$.
%\item Add to $W_R$ all vertices $v$ where cop is on turn and for every edge $(v,w)\in E(H)$ the vertex $w\in W_R$.
%\item Repeat the steps 3 and 4 for $|V(H)|$-times.
%\item If $s\in W_R$ the game $\G$ is robber-win, otherwise the game $\G$ is cop-win.
%\end{enumerate}
%
%Note that each step can be computed in time polynomial in the size of $H$.
It remains to show that the size of $H$ is $2^{O(n)}$. As mentioned in the
introduction, the simplest upper bound on $|V(H)|$ is $2|V(G)|^{c+1}$,
which is unfortunately super-exponential in $n$ if $c$ is close to $n$.
To find a better upper bound, we use the fact that the cops are mutually indistinguishable.
There are at most $|V(G)|$ positions of Robber. Counting the number
of all positions of the cops is the classical problem of putting $c$ indistinguishable
balls into $|V_C|$ bins.
Then, taking into account also whose turn it is and the number of vertices $r_i$,
we get that $|V(H)|$ is bounded by
\[
|V(H)|\le 4|V(G)|{|V_C|+c-1\choose c}\le 4n{n+c-1\choose c} \le 4n2^{n+c-1} = 2^{O(n)}.
\]
Thus the total size of $H$ is $2^{O(n)}$ as well.
The proof for the reverse version of the game is analogous.
\qed\end{proof}

\subsection{Reduction from a formula-satisfying game}

Let us first study the problem after the second move, where both players
have already placed their pawns.
We reduce \GuardPSP\ from the following formula-satisfying game
which was described by Stockmeyer and Chandra~\cite{StoCha} (game $G_2$ in Section~3 of \cite{StoCha}).

A position in the formula-satisfying game configuration is a $4$-tuple $\F=(\tau, F_R(C,R), F_C(C,R), \alpha)$
where $\tau\in\{1,2\}$, $F_R$ and $F_C$ are formulas in 12-DNF
both defined on set of variables $C\cup R$, where $C$ and $R$ are disjoint
and $\alpha$ is the current $(C\cup R)$-assignment. The symbol $\tau$
serves only to differentiate the positions where the first or the second
player is to move.
Player I (II) moves by changing the values assigned to at most one variable in $R$ ($C$);
either player may pass since changing no variable amounts to a ``pass''.
Player I (II) wins if the formula $F_R$ ($F_C$) is true after some move of player I (II).
More precisely, player I can move from $(1,F_R,F_C,\alpha)$ to $(2,F_R,F_C,\alpha')$
in one move if and only if $\alpha'$ differs from $\alpha$ in the assignment given to at most one variable
in $R$ and $F_C$ is false under the assignment $\alpha$; the moves of player II
are defined symmetrically.

Let \FormSAT\ denote the language of all binary encodings of formula-satisfy\-ing game configurations
which are winning for player I.

Note, that if the game continues infinitely, then the starting position is not in \FormSAT,
as player~I did not win. We may as well consider the game won by player~II, which is consistent
with the rules of the cop-game.

The following result is known.

\begin{theorem}[Stockmeyer, Chandra \cite{StoCha}]\label{t:stocha}
\FormSAT\ is an \E-complete language under log-space reduction.
Moreover, $\E\logred\FormSAT$ via length order $n\log n$.
\end{theorem}

Let us first informally sketch the reduction from $\F$ to $\G$,
i.e., simulating $\F$ by an equivalent guarding game $\G$.
The setting of variables is represented by positions of certain cops
so that only one of these cops may move at a time (otherwise the cop-player loses the game),
we call this device a Variable Cell.
The variables in~$C$ are under the control of the cop-player.
The variables in~$R$ are under the control of the robber-player (even though they are also
represented by cops).

In the following text, we construct a lot of gadgets.
First, we precisely define them and then we informally describe how particular gadgets
work, together with prerequisites needed for the gadget to work properly, and with the desired
outcome of each gadget.
Then we describe in detail the whole construction and how all gadgets are connected together.
Finally, we precisely prove that the gadgets (and the whole game) behave as we described earlier.

When describing the features of various gadgets, we will often use the term \emph{binding scenario}.
By binding scenario of a certain gadget (or even the whole game) we mean a flow of the game
that both players must follow: if a player deviates from it, he will lose the game.
The graph will be constructed in such a way, that if the players follow the binding scenario,
they will simulate the formula game $\F$.

Let us define the binding scenario more precisely.

\begin{definition}
Let $\G=(G,V_C,c,S,r)$ be an instance of the guarding game with prescribed positions.
Let $Q=(r_1,S_1,r_2,S_2,\dots,r_k,S_k)$ be a nonempty sequence such that $r_i\to r_{i+1}$ is a valid move of Robber
and $S_i \to S_{i+1}$ is a valid move of cops for every $i$.
We call such $Q$ a \emph{move sequence}. A set $\F$ of move sequences is called a \emph{scenario}.
We say that players \emph{follow the scenario $\F$} if they make a move sequence $Q\in\F$.
We say that the player $P$ \emph{deviates from the scenario $\F$}
if $P$ makes a move not permitted by $\F$, but until his move the players
did follow~$\F$.

We call a scenario $\F$ a \emph{binding scenario} when $\F$
satisfies the following condition:
If the cop-player deviates from $\F$, he will lose in at most ten turns.
If the robber-player deviates from $\F$, he ends in a losing state of the game.
\end{definition}

When describing a particular gadget, we also describe its scenario
and prerequisites which must be satisfied in order for the gadget to work properly.
Later, after we complete the whole construction, we prove that the presented scenarios
are in fact binding, if the prerequisites are satisfied.

There are four cyclically repeating phases of the game,
determined by the current position of Robber.
The binding scenario is as follows:
Robber cyclically goes through the following phases
marked by four special vertices and in different phases he can enter certain gadgets.
(Robber is allowed to pass. However, as we will show, it is never to his advantage.)

\begin{enumerate}
\item ``Robber Move'' ($\var{RM}$): In this step Robber can enter a \emph{Robber Switch} gadget, allowing
him to change setting of at most one variable in $R$.
The values of formula variables are represented by positions of special cops in gadgets called Variable Cells.
\item ``Robber Test'' ($\var{RT}$): In this step Robber may pass through the \emph{Robber Gate} into the protected region $V_C$,
provided that the formula $F_R$ is satisfied under the current setting of variables.
\item ``Cop Move'' ($\var{CM}$): In this step the cop player can change one variable in $C$.
This is realized by a gadget called \emph{Cop Switch}.
At the end of this phase Robber `waits' on vertex $\var{RW}$ for the cop to finish setting the variable.
\item ``Cop Test'' ($CT$): In this step, if the formula $F_C$ is satisfied under the current setting of variables,
the cops are able to block the entrance to the protected region forever (by temporarily leaving
the \emph{Cop Gate} gadget unguarded and sending a cop to block the entrance to $V_C$ that is provided
by the Robber Gate gadgets).
\end{enumerate}

%See Fig.~\ref{f:constr} for the overview of the construction.
%\picture{constr}{The sketch of the construction}{0.7}

To begin with the construction, we start with the basic five vertices
organised to a directed cycle $(\var{RM},\var{RT},\var{CM},\var{RW},CT)$,
we call it the \emph{Basic Cycle}.

\subsection{The Variable Cells}

To maintain the current setting of a variable $x$, we use the Variable Cell gadget,
there will be a separate Variable Cell for every variable.

\picture{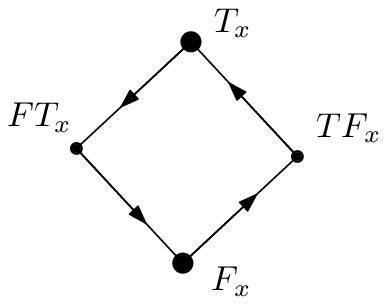}{Variable Cell $\VAR(x)$}{0.3}

\gadget{Variable Cell}
For a variable $x\in C\cup R$ we introduce a \emph{Variable Cell} $\VAR(x)$,
which is a directed cycle $(T_x,\var{FT}_x,F_x,\var{TF}_x)$ (see Fig.~\ref{f:varcell}).

The informal description of the function of $\VAR(x)$ is as follows.
There is one cop (\emph{variable cop}) located in every $\VAR(x)$ and the position of the cop on vertices
$T_x$, $F_x$ represents the boolean values true and false, respectively.
The prescribed starting position of the variable cop is $T_x$ if $\alpha(x)$ is true,
and $F_x$ otherwise.
All the vertices of $\VAR(x)$ belong to $V_C$.

%The Cells are organised into blocks $C$ and $R$ based on the corresponding variable~$x$.

\subsection{Forcing and blocking}\label{s:forceblock}

To make the variable cops truly represent the boolean values and allow only
the changes of these values that follow the rules of the formula game,
we introduce two gadgets that will be used repeatedly.

We say that we \emph{block} a set $S \subseteq V_C$ by a vertex $u$,
when we add to our graph the \emph{Blocker} gadget $\BLOCK(u,S)$.
The Blocker $\BLOCK(u,S)$ consists of a directed path
$(u, p_{u,S}, q_{u,S})$ together with edges $(s,q_{u,S})$ for $s \in S$.
The vertex $p_{u,S}$ is in $V_R$, while $q_{u,S}$ is in $V_C$.
We will call $q_{u,S}$ the \emph{entry vertex} of the blocker.
%$p_{u,S}$ its \emph{forcing vertex}.
See Fig.~\ref{f:blocking} for illustration.

\picture{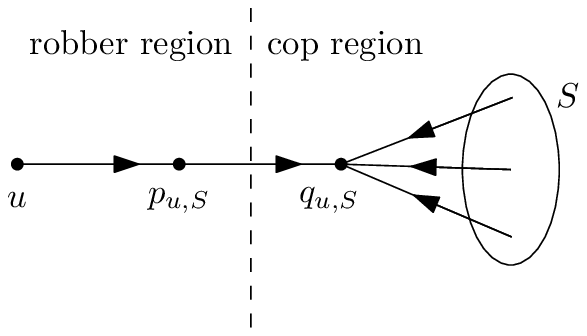}{The Blocker $\BLOCK(u,S)$}{0.5}

Suppose we blocked $S$ by $u$ and the only edges coming to $q_{u,S}$
start in~$S$. If Robber is at $u$, then we must have at least one cop
in the set $S$ or on $q_{u,S}$.

We prove the properties of the Blocker in Lemma~\ref{l:blocking}.

\medskip

A way to prevent cops from moving arbitrarily is called \emph{forcing}
and we will now describe it.
Consider $S \subset V_R$ (we will apply this for $S$ consisting of all the ``normal'' positions of Robber)
and $V \subset V_C$.
We say that we \emph{force $V$ by $S$}, when we add the gadget
$\FORCE(S,V)$ in Fig.~\ref{f:forcing}:
it consists of a vertex $f_{S,V}$, together with edges from each vertex in $S$ to it.
Further, for any $v \in V$ we add a directed path $(f_{S,V}, \bar v, v)$.
All of the new vertices $f_{S,V}$ and $\bar v$ are added to the robber-region $V_R$.

\picture{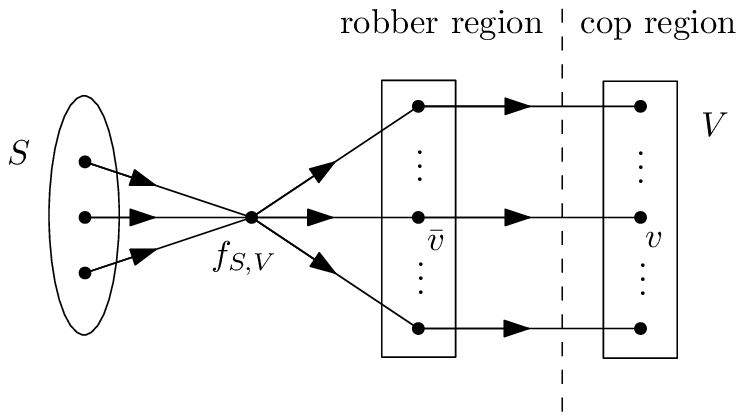}{The Enforcer gadget $\FORCE(S,V)$.}{0.6}

We describe here shortly the desired behaviour.
We say a vertex $u \in V_C$ is \emph{guarded} if there is a cop that can reach $u$
in at most two moves. The forcing gadget ensures that before each Robber's move,
every vertex of~$V$ is guarded.
Later, we will apply the gadget $\FORCE(S,V)$ for $V = \{F_x\}$ and $V = \{T_x\}$
(and for one more similar pair).
The desired outcome is that the variable cop is forced to stay on $T_x$ or $F_x$, unless
another cop helps to guard these two vertices.
(We could ensure the same behaviour using blocker gadgets, but it would lead to a larger
graph -- one with quadratically many edges -- thus we would obtain weaker result
for the running time of algorithms that solve \GuardPSP.)

The properties of forcing are precisely described and proved in Lemma~\ref{l:forcing}.

\subsection{The Robber Gate to $V_C$}

For every clause $\varphi$ of $F_R$, we construct Robber Gate gadget $\ROBG(\varphi)$.
The goal of this gadget is that it allows Robber to enter $V_C$ if and only if
$\varphi$ is satisfied by the current setting of variables.

\picture{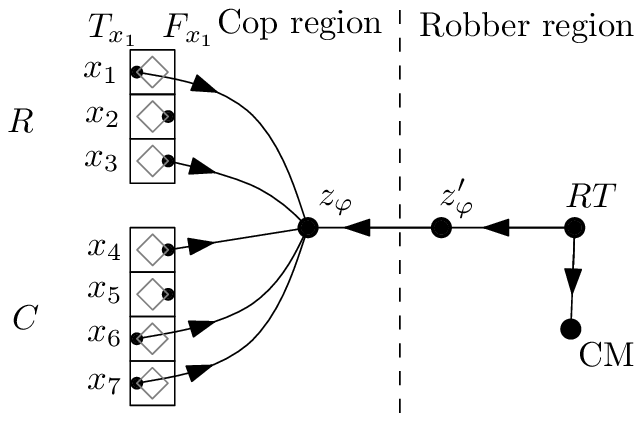}{The Robber Gate $\ROBG(\varphi)$ for the clause
$\varphi=(\neg x_1 \& x_3 \& x_4 \& \neg x_6 \& \neg x_7)$}{0.5}

\gadget{Robber Gate $\ROBG(\varphi)$}
Consider a clause $\varphi=(\ell_1 \& \dots \& \ell_{12})$ where each $\ell_i$ is a literal.
The Robber Gate $\ROBG(\varphi)$ is the Blocker $\BLOCK(\var{RT}, V)$ with
$$
 V = \{T_x \mid \hbox{some $l_i$ is $\neg x$}\}
    \cup
   \{F_x \mid \hbox{some $l_i$ is $x$}\}
$$

For easier notation, we shall use shorter notation
$z'_\varphi := p_{\var{RT},V}$ and $z_\varphi := q_{\var{RT},V}$.
See Fig.~\ref{f:robgate} for illustration.
The vertices $z'_\varphi$ and $\var{RT}$ belong to $V_R$, $z_\varphi$ belongs to $V_C$.

\prereq{Robber Gate $\ROBG(\varphi)$}
Robber stands on $\var{RT}$ and there is exactly one cop in each $\VAR(x)$, $x\in\varphi$, standing
either on $T_x$ or $F_x$.
No other cop can access $z_\varphi$ in one move.
It is the Robber's turn.

\scenario{Robber Gate $\ROBG(\varphi)$}
Robber can reach $z_\varphi$ (thus winning the game) if and only if $\varphi$ is satisfied
under the current setting of variables (given by the positions of cops on Variable Cells) \emph{and} the
cop-player did not win before.
Otherwise, Robber moves to $\var{CM}$.

We remark that if cop-player has won before (by satisfying his formulas in the formula-game -- this
corresponds to moving one of his pawns to vertex $a''$ described in the next section) then there will be a cop that
can reach every $z_\varphi$ in a single move. This renders the Robber Gate unfunctional.

The properties of the Robber Gate are proved in Lemma~\ref{l:robbergate}.

\subsection{The Cop Gate}

For every clause $\psi$ of $F_C$, we will use the \emph{Cop Gate} gadget $\COPG(\psi)$ (see Fig.~\ref{f:copgate}).
If (and only if) $\psi$ is satisfied, $\COPG(\psi)$ allows cops to win the game. According to the rules, the game
will still continue infinitely long, but Robber will never be able to enter $V_C$. This is achieved by
placing a cop to a vertex, from which he controls all vertices $z_\varphi$ in the Robber Gates.

\gadget{Cop Gate $\COPG(\psi)$}
The Cop Gate $\COPG(\psi)$ contains a directed cycle $(a_\psi,a'_\psi,a'',a'''_\psi)$.
Further, $\COPG(\psi)$ contains edges $(a'', z_\varphi)$ for every clause~$\varphi$ of the Robber's formula~$F_R$.
(These edges are common for all clauses~$\psi$.)
Note, that the vertex $a''$ does not have the subscript~$\psi$ -- this vertex is shared
among all clauses~$\psi$. The reason for this is to keep the size of the constructed graph~$\dir G$ small enough.

Let $\psi=(\ell_1 \& \dots \& \ell_{12})$ where each $\ell_i$ is a literal.
We add 12 blockers to our gadget, one for every literal: \\
If $\ell_i=x$ then we add the blocker $\BLOCK(CT, \{T_x, a_\psi, a''\})$ to $\COPG(\psi)$. \\
If $\ell_i=\neg x$ then we add the blocker $\BLOCK(CT, \{F_x, a_\psi, a''\})$ instead.

The directed 4-cycle belongs to $V_C$, the blocker is as described above.

\picture{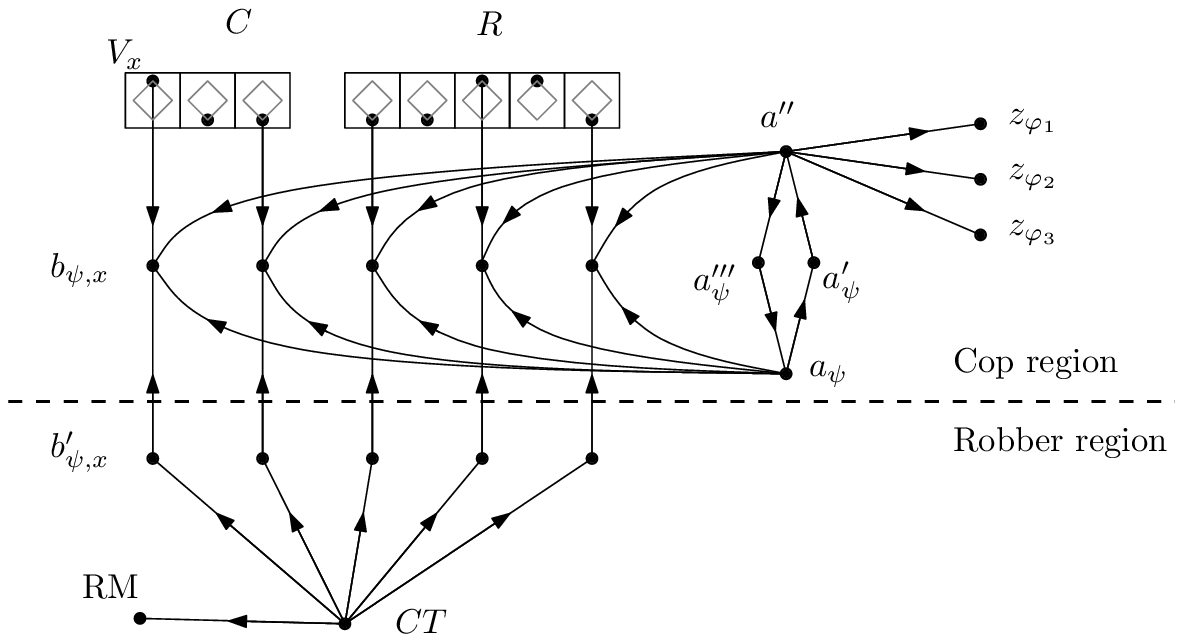}{The Cop Gate $\COPG(\psi)$ for $\psi = x_1 \& \neg x_3 \& \neg y_1 \& y_3 \& \neg y_5$,
($x_i$ are the cop's variables, $y_i$ the Robber's ones). Note that we only add vertices corresponding
to the variables that appear in the clause $\psi$.}{0.9}

\prereq{Cop Gate $\COPG(\psi)$}
There is one cop at the vertex $a_\psi$ (we call him Arnold)
and there is exactly one cop in each $\VAR(x)$, $x\in\psi$, standing on either $T_x$ or $F_x$.
Robber is at the vertex $CT$ and no other cop can access $\COPG(\psi)$ in less than three moves.
It is the cop's turn.

\scenario{Cop Gate $\COPG(\psi)$}
Arnold is able to move to $a''$ (and therefore block all
the entrances $z_\varphi$ forever) without permitting Robber to enter $V_C$
if and only if $\psi$ is satisfied under the current setting of variables (given by the position
of cops in the Variable Cells).

Note that while Arnold will never move voluntarily back from~$a''$ to $a_\psi$,
Robber may force him using the forcing gadget that we will add
in Section~\ref{s:bigpicture} in step~\ref{item:force} of the construction.
Therefore, the vertex $a'''_\psi$ is needed.

The properties of the Cop Gate are proved in Lemma~\ref{l:copgate}.

\subsection{The Commander gadget}\label{s:commander}

When changing variables of $R \cup C$, we have to make sure that at most one
variable is changed at a time. We already prevented all variable cops from moving at all
by means of forcing, see Section~\ref{s:forceblock} and Lemma~\ref{l:forcing}.

Next we describe a gadget that allows one of them to move: the gadget Commander (see Fig.~\ref{f:commander}).

\gadget{Commander $\COM$}
It consists of vertices $\{G_x, GF_x, GT_x \mid x\in R \cup C\}\cup\{\var{HQ}\}$.
For each $x \in R \cup C$ there is directed 2-cycle $(\var{HQ},G_x)$ and directed
paths $(G_x, GF_x, F_x)$ and $(G_x, GT_x, T_x)$.
All of the vertices belong to $V_C$.

\picture{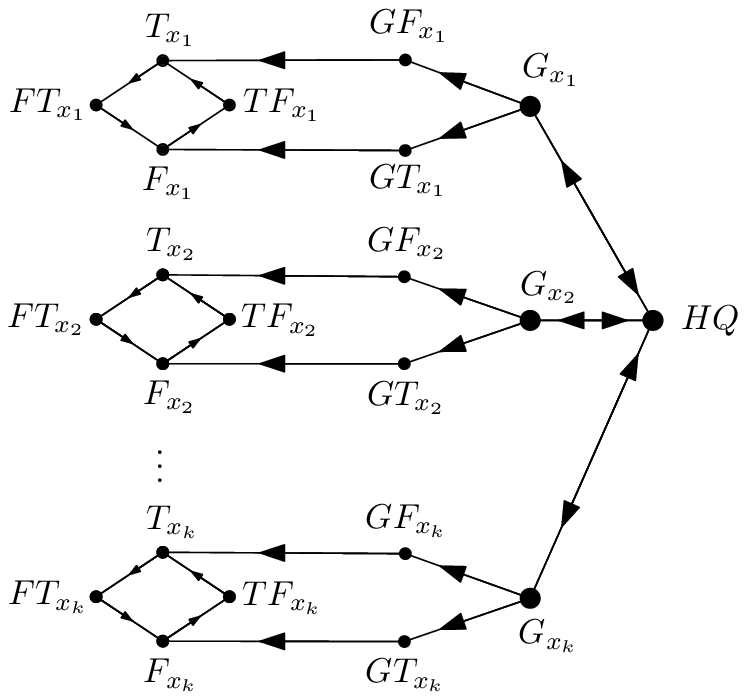}{The Commander gadget $\COM$ including Variable Cells $\VAR(x_1),\dots,\VAR(x_k)$.}{0.7}

There is one cop, Commander, whose prescribed starting position is the vertex $\var{HQ}$.
In Section~\ref{s:bigpicture} we will add some extra edges that will force Commander to
stay at $\var{HQ}$ most of the time.

The desired function of $\COM$ is as follows:
in order to move a variable cop from $T_x$ to $F_x$ or vice versa, Commander
moves to $G_x$ to temporarily guard vertices $T_x$ and $F_x$.
This will happen in two stages of the game: when Robber sets variables, he will pick a particular
$x \in R$ and make Commander move to~$G_x$.
When the cop-player moves, Robber will only ensure that Commander moves to a $G_x$ with $x \in C$.

\gadget{Robber Switch $\ROBS(y)$}
Let $y\in R$. We describe a gadget that allows Robber to change
the value of the variable $y$, i.e., to force the variable cop in $\VAR(y)$
to move from $T_y$ to $F_y$ or vice versa.

The Switch of Robber's variable $y$ ($\ROBS(y)$) consists of
a directed path $(\var{RM},Sw_y,\var{RT})$, edges $(\var{RM}, \var{HQ})$, $(\var{RT}, \var{HQ})$, $(Sw_y,G_y)$,
$(\var{RM},\var{RT})$
and a Bloc\-ker $\BLOCK(Sw_y, \{\var{FT}_y, \var{TF}_y\})$
(see Fig.~\ref{f:switch}).

The new vertices $Sw_y$ belong to $V_R$, the Blocker is described earlier.

\picture{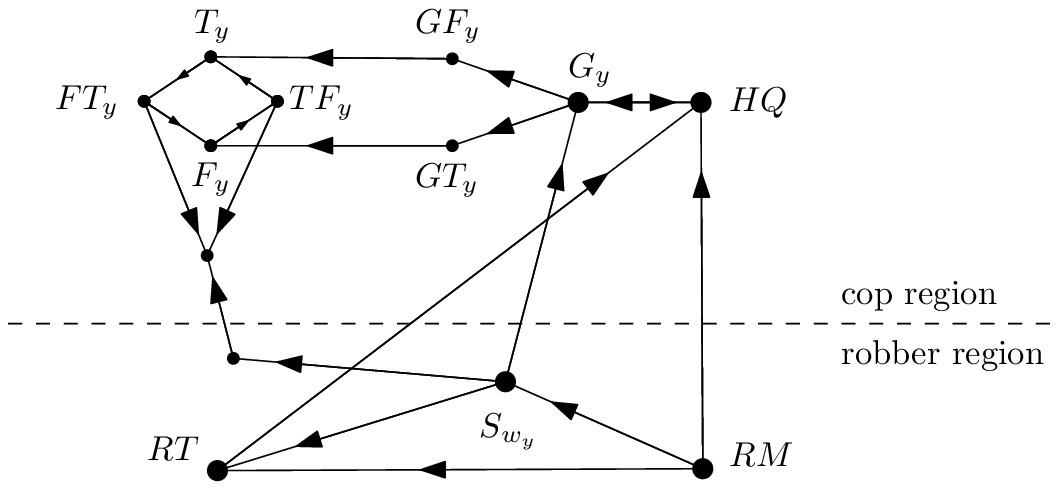}{The Robber Switch gadget $\ROBS(y)$ together with the
Variable Cell $\VAR(y)$ and a part of the Commander gadget $\COM$.}{0.9}

\prereq{Robber Switch $\ROBS(y)$}
Robber is at $\var{RM}$, the cop in $\VAR(y)$ is
either on $T_y$ or $F_y$, Commander is at $\var{HQ}$ and no other cop can access any vertex of $\VAR(y)$.
It is the Robber's turn.

\scenario{Robber Switch $\ROBS(y)$}
By entering the vertex $Sw_y$
Robber forces Commander to move to $G_y$, and the variable cop to
enter $\var{FT}_y$ or $\var{TF}_y$.
Robber then enters the vertex $\var{RT}$, the variable cop finishes his
move to $T_y$ (or $F_y$) and at the same time Commander returns to $\var{HQ}$.
If Robber decides not to change any variable, he may
go directly to $\var{RT}$; in such case the cops do not move in their next turn.

\smallskip

The properties of the Robber Switch are proved in Lemma~\ref{l:robs}.

\smallskip

\gadget{Cop Switch $\COPS$}

The $\COPS$ gadget consists of a Blocker $\BLOCK(\var{CM},\{G_x \mid x \in C\}$.

\prereq{Cop Switch $\COPS$}
Robber is on $\var{CM}$, all variable cops are at $T_x$ or $F_x$ and Commander at $\var{HQ}$.
It is the cop's turn.

\scenario{$\COPS$}
The cop-player decides which variable $x \in C$ to change.
The Commander moves to $G_x$, and the variable cop in $\VAR(x)$
moves to $\var{TF}_x$ or $\var{FT}_x$. Next, Robber moves to $\var{RW}$.
Commander moves back to $\var{HQ}$ and the cop in $\VAR(x)$
finishes his move to $F_x$ or $T_x$.
Finally, Robber moves to $CT$.
Alternatively, if the cop-player does not want to change any variable,
all variable-cops stay put, and just Commander moves to some $G_x$ and
back to $\var{HQ}$.

\subsection{The big picture: putting the gadgets together}\label{s:bigpicture}

Let us consider an instance of the formula game $\F=(\tau, F_R(C,R), F_C(C,R), \alpha)$.
We now proceed by putting all gadgets together in order to construct the instance
$\G=(\dir G,V_C,c, S, r)$ of \GuardPSP.

When different gadgets contain a vertex with the same name, this means that this vertex
is shared by these gadgets (this happens often for vertices $\var{RM}$, $\var{RT}$, $\var{CM}$, $CT$, $\var{HQ}$, $a''$ but also for others).
To make the construction clearer, we define the following sets:
\begin{align*}
  \calT &= \{ T_x \mid x \in R \cup C\} \\
  \calF &= \{ F_x \mid x \in R \cup C\} \\
  \calA &= \{ a_\psi \mid \psi \in F_C \} \\
\end{align*}

The order of construction steps is as follows.

\begin{enumerate}

\item\label{item:basiccycle} Construct the Basic Cycle.

\item\label{item:varcell} For every variable $x\in C\cup R$ construct the Variable Cell gadget $\VAR(x)$.

\item\label{item:robbergate} For every clause $\varphi$ of $F_R$ construct the Robber Gate gadget $\ROBG(\varphi)$.

\item\label{item:copgate} For every clause $\psi$ of $F_C$ construct the Cop Gate gadget $\COPG(\psi)$.

\item\label{item:commander} Construct the Commander gadget $\COM$.

\item\label{item:robs} For every variable $y\in R$ construct the Robber Switch $\ROBS(y)$.

\item\label{item:cops} Construct the Cop Switch $\COPS$.

\item\label{item:force}
  Let $S_1 = \{ \var{RM}, \var{RT}, \var{CM}, \var{RW}\} \cup \{ Sw_y \mid y \in R \}$ and
  $S_2 = S_1 \cup \{ CT \}$. Add to the graph gadgets
  $\FORCE(S_2, \calT)$, $\FORCE(S_2, \calF)$,
  $\FORCE(S_1, \calA)$, and $\FORCE(S_1, \{a''\})$.
  We call the vertices $f_{S_2, \calT}$, $f_{S_2, \calF}$, $f_{S_1, \calA}$, and
  $f_{S_1, \{a''\}}$ of these gadgets the \emph{forcing vertices}.

\item\label{item:edges} Add directed edges to $\var{HQ}$ from $\var{RM}$, $\var{RT}$, $\var{RW}$, $CT$.
\end{enumerate}

%\includegraphics{graph.jpg}
%\includegraphics[page=2,scale=0.5]{graphxoj.pdf}

%Note that the steps \ref{item:enforfirst}--\ref{item:enforlast} must be done as the last ones,
%because they depend on the previous construction elements.

\medskip

The prescribed starting position $r$ of Robber is the vertex $\var{RM}$.
We define the starting positions $S$ of the cops as follows:
\begin{itemize}
\item For each variable $x\in C\cup R$, the corresponding variable cop
starts at the vertex $T_x$ if $\alpha(x)$ is true and
at the vertex $F_x$ if $\alpha(x)$ is false.
\item Each $\psi$-Arnold starts at the corresponding vertex $a_\psi$.
\item The cop Commander starts at the vertex $\var{HQ}$.
\end{itemize}

The number $c$ of cops is thus set as $c=|C\cup R|+|F_C|+1$, where
$|F_C|$ denotes the number of clauses in~$F_C$.

\medskip

We also briefly show that the construction of $\G$ based on $\F$ can be done in log-space.
A log-space reduction is a reduction computable by a deterministic Turing machine
(with read-only input tape, a working tape, and output tape) using logarithmic space.

Conceptually, this means that while producing encoding of the instance $\G$ to the output tape
the reducing Turing machine can keep (stored on the work tape)
a constant number of pointers into the input tape
along with a constant number of integers holding at most polynomially large values,
and a logarithmic number of fixed size integers.

\begin{lemma}\label{l:logspace}
Let $\F$ be an instance of the formula game and let $\G$
be an instance of \GuardPSP\ constructed via the above construction from $\F$.
Then the reducing Turing machine is log-space bounded and $|\G|=O(|\F|\log |F|)$
where the instance sizes are measured in bits.
\end{lemma}

\begin{proof}
In order to show the log-space reducibility of our construction of $\G$,
we need to observe that the whole construction can be performed by an algorithm
reading the encoding of the formula game $\F$ and using only constant number
of work variables holding polynomially large values.

We go through the construction and briefly discuss the space needed for construction.

\begin{itemize}
\item Step \ref{item:basiccycle} of the construction produces a fixed-size gadget.
\item Steps \ref{item:varcell}, \ref{item:commander}, and \ref{item:robs}
can be realised by an iteration over the set of variables,
each iteration producing a fixed-size gadget. Clearly, a constant number of log-space integers suffices.
\item Steps \ref{item:robbergate} and \ref{item:copgate} can be realised by two
nested iterations over the set of clauses and the set of variables.
Again, a constant number of log-space integers suffices.
\item Step \ref{item:force}
can be done by several iterations (over the set of variables and over the set of clauses).
Important thing to note is that the gadget $\FORCE(S,V)$ has
size (the number of vertices and edges) $O(|S|+|V|)$.
\end{itemize}
The output of the cops and Robber starting positions can be done by an iteration
over the set of clauses and variables.
Therefore, only a constant number of log-space integers is needed during the reduction.

Let $f$ be the number of clauses of $\F$ and let $p$ be the number of variables of~$\F$.
Then the number of bits of $\F$ is $\Omega(f\log p + p) = \Omega(f+p)$ as we need $\log p$ bits to encode
the variable identifiers, each clause contains at most 12 literals and the initial
assignment of variables takes $p$ bits.

Next, we estimate the size of the instance $\G=(\dir G,V_C,c,S,r)$, based on our construction,
starting with the number of vertices and edges of $\dir G$.

Each $\VAR(y)$ has a constant size, all Variable Cells thus have total size $O(p)$.
Each $\ROBG(\varphi)$ has a constant size, all Robber Gates thus have total size $O(f)$.

The Cop Gate gadgets contain a shared part of size $O(f+p)$
and the non-shared part of each $\COPG(\psi)$ has a constant size.
Therefore, all Cop Gate gadgets have total size $O(f+p)$.
The size of the Commander gadget is clearly $O(p)$.

Each $\ROBS(y)$ has a constant size, so the total size of all Robber Switches is $O(p)$.
The size of $\COPS$ is clearly $O(p)$.
In the last nontrivial step, step~\ref{item:force}, $O(p+f)$ vertices and edges are added to $\dir G$.

As $|V(\dir G)|+|E(\dir G)| \le O(p+f)$, we need
$$
O((p+f)\log(p+f)) = O(|\F| \log(|\F|) )
$$
bits to encode $\dir G$.
\qed\end{proof}

Note that we may get a linear bound on the size of $\G$, if we assume fairly reasonable condition on the formula game:
all variables are used and no clause repeats within a formula.

\subsection{Proof of correctness}

Let us start by a simple observation.
While the rules allow Robber to pass (skip a move),
we may assume that he never uses it.
For if Robber does not move, the cops may do the same and repeating
this indefinitely means a win for cops.
From now on, we shall assume Robber never passes.

Another assumption, that will be somewhat harder to prove:
%\begin{center}
\begin{adjustwidth}{1cm}{1cm}
Assumption A: the cop-player restricts to strategy where either every $\psi$-Arnold is at his starting vertex~$a_\psi$,
or exactly one of them moves to~$a'_\psi$,~$a''$ (when Robber entered vertex~$CT$), and possibly back to $a'''_\psi$ and $a_\psi$.
\end{adjustwidth}
%\end{center}
In Theorem~\ref{t:prescribedgame} we will show that this restriction does not change the outcome of the game,
it will make it easier to analyze, though.

Before getting to the individual lemmas, we describe the main idea: if the robber-player
succeeds in satisfying formula~$F_R$, Robber will be able to enter the cop-region by means of Robber Gate.
If the cop-player will be faster with satisfying his formula~$F_C$, he will be able to send one
Arnold to~$a''$. The resulting state of game will be cop-win -- there is an easy strategy
(described in Lemma~\ref{l:arnoldwin}) that will make sure Robber will never enter the cop-region.

Recall that we call a vertex $v \in V_C$ \emph{guarded} if there is a cop that can reach $v$
in at most two moves.
Further, a \emph{set} $V \subseteq V_C$ is \emph{guarded} if every vertex of~$V$ is guarded
and each one is guarded by a different cop. Formally, we require that there is
a set $U$, $|U|=|V|$, of cops and a 1-1 mapping $m$ between $U$ and $V$
such that for each $c\in U$ the vertex $m(c)$ is guarded by the cop~$c$

\begin{lemma}\label{l:forcing}
Let $S$, $V$ be sets such that the gadget $\FORCE(S,V)$ is part of $\dir G$.
Let us consider the game state where Robber is in $S$ and it is Robber's turn.
\begin{enumerate}
\item
Let there be an unguarded vertex of~$V$. Then this is a robber-win state.
\item
If the set $V$ is guarded and Robber enters the vertex $f_{S,V}$, the game state will be a cop-win.
\item
If the set~$V$ is not guarded and there is no vertex $x \in V_C$ that has more than one outneighbor
in~$V$ then this is a robber-win state.
\end{enumerate}
\end{lemma}

\begin{proof}
Suppose $v$ is unguarded.
Then Robber moves along the path $(f_{S,V}, \bar v, v)$
and no cop can stop him.

On the other hand, suppose the set~$V$ is guarded and mapping~$m$ is defined as above.
If Robber decides to move to $f_{S,V}$, than each cop~$c$ guarding vertex $m(c) \in V$ moves towards it
(or stays in place, if they already are at it).
Thus, the cops can reach all vertices of~$V$ before Robber, and Robber
is stuck in~$f_{S,V}$ or some $\bar v$ for the rest of the game.

For the third part, suppose that Robber moves to $f_{S,V}$ and the cops can somehow prevent him
from winning. Then for every $v \in V$ the cops have a response to Robber's move to $\bar v$, namely,
there is a cop that can move to~$v$ in one step, denote him by $d(v)$. Due to the extra condition
in the third part, mapping $d$ is injective, its inverse is the required mapping~$m$.
\qed\end{proof}

We now define the notion of \emph{normal positions} of certain pawns.
\begin{definition}\label{def:normal}
We say that
\begin{itemize}
  \item Robber is in normal position if he is on $S_2$ (which was previously defined as
    $\{ \var{RM}, \var{RT}, \var{CM}, CT, \var{RW}\} \cup \{ Sw_y \mid y \in R \}$),
  \item Commander is in normal position if he is on $\var{HQ}$ or on some $G_x$,
  \item variable cops are in normal position if they are in some $\VAR(x)$,
  \item $\psi$-Arnold is in normal position if he is at $a_\psi$ or $a''$.
\end{itemize}
\end{definition}

If one of the players leaves these positions, the game will
be decided quickly. This is obvious for Robber, for if he is not in the
normal position, he is no longer in the strongly connected part of the graph.
So, he either reaches $V_C$ within
two moves, or he will be unable to move for the rest of the game.
In the following lemmas we deal with the rest of the pawns.

\begin{lemma}\label{l:HQ}
Consider a game state where Robber is in $S_2$ (see Definition~\ref{def:normal})
and it is the Robber's turn.
If Commander is not on $\var{HQ}$ or on some~$G_x$, then this is a robber-win state.
\end{lemma}

\begin{proof}
If Commander is elsewhere, he will not be able to reach $\var{HQ}$ again
(by the construction). If Robber is in~$S_2$, he will be able to get to $\var{RM}$
and from there to enter $\var{HQ}$.
\qed\end{proof}

\begin{lemma}\label{l:normalposvarcops}
Consider a game state where Robber is in $S_2$ (see Definition~\ref{def:normal})
and it is the Robber's turn.
If there is an $x \in R \cup C$ so that there is no cop on $F_x$, $T_x$, or $G_x$,
then this is a robber-win state. Otherwise, if Robber enters vertex
$f_{S_2,\calF}$ or $f_{S_2,\calT}$, he loses.
\end{lemma}

\begin{proof}
By Lemma~\ref{l:HQ}, Commander must stay on~$\var{HQ}$ or some~$G_x$.
Consequently, there is at most one cop in each $\VAR(x)$ and no cop is
at vertices $GF_x$, $GT_x$. Thus, to guard both $T_x$ and $F_x$, we need a cop
at either $G_x$, $T_x$, or $F_x$ (Lemma~\ref{l:forcing}).
The second statement follows again from Lemma~\ref{l:forcing}, as
no cop will be guarding two of the vertices $T_x$, neither two of the vertices $F_x$.
\qed\end{proof}

\begin{lemma}\label{l:normalposarnolds}
\begin{enumerate}
\item Consider a game state where Robber is in~$S_2$ (see Definition~\ref{def:normal})
and it is the Robber's turn.
Then all Arnolds are in $X=\bigcup_\psi \{a_\psi, a'_\psi, a'''_\psi\} \cup \{a''\}$,
or it is a robber-win state.
%If there is a $\psi \in F_C$ so that there is no cop on $a_\psi$, $a''$ or $a'''_\psi$, then it is a robber-win state.
\item If every $\psi$-Arnold is on $a_\psi$, with possible exception of one,
      who is at $a''$, then Robber may not enter vertices $f_{S_1,\calA}$, $f_{S_1,\{a''\}}$.
\item If Robber is in~$S_1$, some Arnold in~$a'_\psi$ and it is the Robber's turn,
  then it is a robber-win state.
\end{enumerate}
\end{lemma}

\begin{proof}
Again, a simple consequence of Lemma~\ref{l:forcing}: there is no edge leading towards~$X$, thus
if some Arnold leaves~$X$, there will be always less than~$|\calA|$ cops in~$X$. When (now or in the next move)
Robber is in~$S_1$, he can use the gadget $\FORCE(S_1, \calA)$ according to the third part of Lemma~\ref{l:forcing}.
This proves the first part and the third part.
By the second part of Lemma~\ref{l:forcing} follows the second part.
\qed\end{proof}

\begin{lemma}\label{l:blocking}
Let a set $S$ be blocked by a vertex $u$. (That is, the Blocker gadget
$\BLOCK(u,S)$ is a part of the constructed graph $\dir G$.)
%FIXME: opravit zneni
Suppose that no other edges leading to $q_{u,S}$ are part of~$\dir G$.
Suppose Robber is on $u \in S_2$ and it is his
turn to move. Then there must be at least one cop in the set $S$, otherwise the
cop-player loses the game.
If there is at least one cop in $S$, Robber may not enter $p_{u,S}$, otherwise
he loses the game.
\end{lemma}

\begin{proof}
By Lemma~\ref{l:HQ}, \ref{l:normalposvarcops} and \ref{l:normalposarnolds},
there is no cop in~$q_{u,S}$ -- all cops must be at vertices distinct from
all vertices of form~$q_{u,S}$.
This means that Robber can reach this vertex in two moves, and the
cops cannot stop him.
On the other hand, if Robber moves to $p_{u,S}$ while there is
a cop on $S$, then this cop moves to $q_{u,S}$, blocking Robber
forever.
\qed\end{proof}

\begin{lemma}\label{l:beginning}
At the beginning of the game $\G$,
the prerequisites of the Robber Switches $\ROBS(y)$, $y\in R$ are satisfied.
\end{lemma}

\begin{proof}
By the construction, Robber is on~$\var{RM}$ and the variable cops for
each variable gadget $\VAR(x)$ are either on $T_x$ or $F_x$. Commander is on~$\var{HQ}$.
No other cop can ever reach $\VAR(x)$.
\qed\end{proof}

\begin{lemma}\label{l:robs}
Consider a game state satisfying the prerequisites of all Robber switch\-es $\ROBS(y)$, $y\in R$.
Then the scenarios of all $\ROBS(y)$, $y\in R$ are binding scenarios.
Moreover, after execution of these binding scenarios, the prerequisites of all
Robber Gates $\ROBG(\varphi)$, $\varphi\in F_R$ are satisfied, unless there is an Arnold
at~$a''$.
\end{lemma}

\begin{proof}
If Robber moves to $Sw_y$, Commander must move to $G_y$, or Robber enters $V_C$ there.
(By the construction, no other cop than Commander can be at $G_y$.)

Also, by Lemma~\ref{l:blocking}, the variable cop in $\VAR(y)$ must move to $\var{TF}_y$ or $\var{FT}_y$.
By Lemma~\ref{l:forcing} and~\ref{l:blocking}, Robber may not enter the forcing vertices, and thus continues
to~$\var{RT}$. After this, Commander must move back to $\var{HQ}$, thus the variable cop in $\VAR(y)$ to
$T_y$ or $F_y$ (Lemma~\ref{l:normalposvarcops}).
If Robber moves directly to~$\var{RT}$, no cop may move (by Lemma~\ref{l:forcing} and~\ref{l:blocking}).
\qed\end{proof}

\begin{lemma}\label{l:robbergate}
Consider a game state satisfying the prerequisites of all Robber Gates $\ROBG(\varphi)$, $\varphi\in F_R$.
Then the scenarios of all Robber Gates $\ROBG(\varphi)$, $\varphi\in F_R$ are binding scenarios.
Moreover, after execution of these binding scenarios, the prerequisites of the Cop Switch
are satisfied.
\end{lemma}

\begin{proof}
If no other cop may access $z_\varphi$, Lemma~\ref{l:blocking} applies:
Robber may enter $z_\varphi$ if and only if there is no cop in the set~$V$ (see the definition
of Robber Gates). By the construction, this happens precisely if $\varphi$ is satisfied by the current setting of the variables.
If he may not enter $z_\varphi$, the only non-losing move for him is to move to $\var{CM}$.
If the vertex $a''$ is occupied (the cop-player has ``won'' already), Robber may not
enter $z_\varphi$ either and moves to $\var{CM}$.

After the Robber's move to $\var{CM}$, the prerequisites of the Cops Switch are satisfied, as the cops did not move yet.
\qed\end{proof}

\begin{lemma}\label{l:cops}
Let us consider a game state satisfying the prerequisites of the Cop Switch.
Then the scenario of the Cop Switch is a binding scenario.
Moreover, after execution of it, the prerequisites of all Cop Gates $\COPG(\psi)$, $\psi\in F_C$,
are satisfied.
\end{lemma}

\begin{proof}
When Robber stands on $\var{CM}$, Commander may leave $\var{HQ}$, as the edge $(CM,HQ)$
is not part of~$\dir G$. In fact, Commander has to move to a
vertex $G_x$ where $x \in C$ is a cop-player's variable because of the Cop Switch gadget.
If the cop-player decided to switch $x$,
the variable cop in $\VAR(x)$ moves also in the first move to $\var{FT}_x$ or $\var{TF}_x$.
Next, Robber moves to $\var{RW}$, after which Commander must return to $\var{HQ}$ and thus
the cop in $\VAR(x)$ finishes his move to $T_x$ or $F_x$ (or keeps staying in place).
Finally, Robber moves to $CT$, which leads to a game state satisfying all prerequisites of all Cop Gates.
\qed\end{proof}

\begin{lemma}\label{l:arnoldwin}
Consider a game state where $\psi$-Arnold is on the vertex $a''$ for some clause $\psi$,
other cops are in normal positions, and Robber is in $S_2$ (see Definition~\ref{def:normal}).
Then this is a cop-win state.
\end{lemma}

\begin{proof}
As long as Robber is in normal position, the cop-player does not move with any
of his cops until he is forced to:
one Arnold stays at $a''$, the others at their initial positions -- even if another
clause in $F_C$ becomes satisfied. The variable-cops and Commander follow the scenarios
of the Cop-Switch and Robber-Switch. Robber may force the Arnold at~$a''$ to move back to his
original position by using the Enforcer gadget~$\FORCE(S_1,\calA)$; however, he will be trapped afterwards.

It is easy to check that Robber may not enter the cop-region,
as in binding scenarios he may do this only in some Robber Gate. But if he moves to some vertex $z'_\varphi$,
the Arnold at $a''$ moves to $z_\varphi$, thus blocking Robber forever.
\qed\end{proof}

\begin{lemma}\label{l:copgate}
Consider a game state satisfying the prerequisites of all Cop Gates $\COPG(\psi)$, $\psi\in F_C$.
Then the scenarios of all $\COPG(\psi)$, $\psi\in F_C$, are binding scenarios.
Moreover, after execution of these binding scenarios, the prerequisites of all Robber Switches $\ROBS(y)$, $y\in R$,
are satisfied.
\end{lemma}

\begin{proof}
If $\psi$ is satisfied by the current setting of variables in the variable cells,
then we do not need the $\psi$-Arnold to protect the entry vertex of any of the blockers,
thus he may move to $a'_\psi$.
If Robber enters some of the Blockers, then he loses by Lemma~\ref{l:blocking}.
On the other hand, if $\psi$ is not satisfied, and $\psi$-Arnold leaves $a_\psi$,
then the cop-player loses: again by Lemma~\ref{l:blocking}.

Next, Robber moves to $\var{RM}$, which forces all vertices $a_\psi$, thus the $\psi$-Arnold
either stays at $a_\psi$ or, if he left it already, finishes his move to $a''$.

Therefore, the scenarios of the Cop Gates are binding.
It is easy to see that the prerequisites of all Robber Switches are satisfied at the end of the binding scenario.
Using Lemma~\ref{l:arnoldwin} we justify that if some Arnold moves to $a''$, it is a win for the cop-player.
\qed\end{proof}

We have described the complete reduction from the formula game $\F$
to the guarding game $\G$ with prescribed starting positions,
and characterised the properties of various stages of $\G$.
It remains to show that both games have the identical outcome;
we also list some properties of the constructed graph that will be
useful later.

\begin{theorem}\label{t:prescribedgame}
%For every instance of formula satisfying game $\F=(\tau, F_C(C,R), F_R(C,R), \alpha)$ there exists
For every instance of formula satisfying game $\F=(\tau, F_C, F_R, \alpha)$ there exists
a guarding game $\G=(\dir G,V_C,c, S, r)$, $\dir G$ directed, with a prescribed starting
positions such that
player I wins $\F$ if and only if the robber-player wins the game $\G$.
Moreover, the following properties are satisfied:
\begin{enumerate}
  \item The size of the instance $\G$ is $|\G|=O(|\F| \log |\F|)$.
  \item There is no oriented edge $e\in E(\dir G)$ with the endpoint in $r$.
  \item There is at most one cop standing at each vertex $v\in V(\dir G)$.
  \item Consider the set $B=\{u\in V_R \mid \exists v\in V_C,\, (u,v)\in E(\dir G)\}$ of the ``border'' vertices.
        The out-degree of each $u\in B$ is exactly 1.
\end{enumerate}
\end{theorem}

\begin{proof}
For every formula game $\F$ we have described in Section~\ref{s:bigpicture}
a construction of guarding game $\G=(\dir G,V_C,c, S, r)$ with prescribed starting positions.

First, we justify Assumption~A. Enforcer gadget~$\FORCE(S_1,\calA)$ does not allow any Arnold
to move from his starting position unless Robber is on~$CT$. When Robber is on~$CT$, Lemma~\ref{l:copgate}
describes the condition under which $\psi$-Arnold may move from~$a_\psi$ to~$a'_\psi$ and further to~$a''$:
it is precisely when clause~$\psi$ is satisfied. The cop-player may refrain from moving more then one Arnold to~$a''$
even if more then one clause in~$F_C$ is satisfied, and he may also avoid moving another Arnold to~$a''$, if one is already there.
This will not change the outcome of the game as (because of Lemma~\ref{l:arnoldwin}), one Arnold on~$a''$ already
means the game is won by the cop-player.

The above proves that using Assumption~A does not change the outcome of the
game. If there is an Arnold at~$a''$, the cop-player has won. Otherwise,
all Arnolds are at their starting positions; this will be important
in Lemma~\ref{l:robbergate} and~\ref{l:copgate}.

By Lemma~\ref{l:beginning},
the assumptions of Lemma~\ref{l:robs}, are satisfied at the beginning of
the game. By Lemma~\ref{l:robs}, the assumptions of
Lemma~\ref{l:robbergate} are satisfied when Robber moves to~$RT$.
Lemma~\ref{l:robbergate} in turn that ensures we can apply Lemma~\ref{l:cops} in
the next step, and similarly we continue with Lemma~\ref{l:copgate} which
closes the cycle by guaranteeing the assumptions of Lemma~\ref{l:robs}.

This altogether means that $\G$ precisely simulates the game $\F$:
In the Robber Move phase (Lemma~\ref{l:robs}) $\G$ imitates setting of
the variables $R$, which is due to Lemma~\ref{l:robbergate} followed in the Robber Test phase
by test if the robber-player wins.
By Lemma~\ref{l:cops} in the following Cop Move phase the game $\G$ imitates setting of
the variables $C$, which is then due to Lemma~\ref{l:copgate} followed in the Cop Test phase
by a test if the cop-player wins.
The process then repeats.

We have thus found the desired game $\G$ and proved its equivalence with $\F$.
The first property is proved in Lemma~\ref{l:logspace}.
The second property is in fact not true, but it is very easy to modify the construction,
so that it is. We create a vertex $r'$ and connect it by an edge to all out-neighbours of~$r$.
We redefine the starting position of Robber as~$r'$. It is immediate, that this modification
does not change the outcome, and $r'$ has no incoming edge.

The third and the fourth items are true by the construction.
\qed\end{proof}

\subsection{Forcing the starting positions and proof of the main theorem}
\label{s:mainproof}

Next we prove that we can modify our current construction so that it fully conforms
to the definition of the guarding game on a directed graph, i.e., without
prescribing the starting positions.

\begin{lemma}\label{l:force}
Let $\G=(\dir G,V_C,c, S, r)$ be a guarding game with a prescribed starting positions,
such that it satisfies the following initial conditions:
\begin{enumerate}
\item There is no oriented edge $e\in E(\dir G)$ with the endpoint in $r$.
\item There is at most one cop standing at each vertex $v\in V(\dir G)$.
\item Consider the set $B=\{u\in V_R \mid  \exists v\in V_C,\, (u,v)\in E(\dir G)\}$ of the ``border'' vertices.
The out-degree of each $u\in B$ is exactly 1.
\end{enumerate}
Then there exists a guarding game $\G'=(\dir G',V'_C,c')$, $\dir G\subseteq \dir G'$, $V_C\subseteq V'_C$ such that
the following holds:
\begin{enumerate}
\item The robber-player wins $\G'$ if and only if the robber-player wins the game $\G$.
\item If the robber-player does not place Robber on $r$ in his first move, the cops win.
\item If the cop-player does not place the cops to completely cover $S$ in his first move, he will lose.
\item The construction of $\G'$ can be made in log-space.
\item $|\G'| = O(|\G|)$
\end{enumerate}
\end{lemma}

\begin{proof}
Let $m=|\{v\in V_C \mid  (u,v)\in E(\dir G),\, u\in V_R\}|$ be the number of vertices from $V_C$
directly threatened (i.e., in distance 1) from robber-region.

Let us define the graph $\dir G'=(V',E')$ such that
$V'=V(\dir G)\cup\{r\}\cup T$ where $T=\{t_1,\dots,t_m\}$ is the set of new vertices and
$E'=E(\dir G)\cup \{(r,v) \mid  v\in T\cup S\}$.
Consider the game $\G'=(\dir G',V'_C,c')$ where $V'_C=V_C\cup T$ and $c'=c+m$.
See Fig.~\ref{f:startforce} for illustration.

\picture{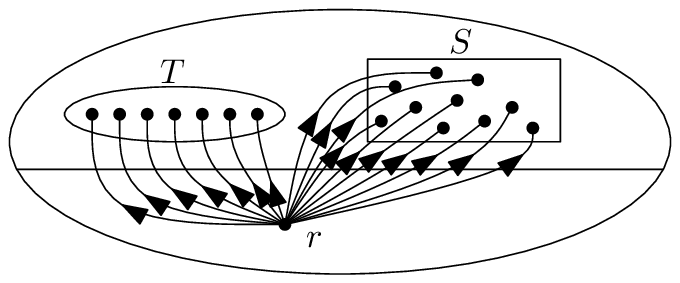}{Forcing starting positions}{0.6}

Suppose that the robber-player places Robber in the first move to some vertex $v\in V_R\setminus\{r\}$.
Then there are $m$ vertices in $V_C$ directly threatened by edges going from $V_R$ and because
we have at least $m$ cops available, the cops in the second move can occupy all
these vertices and prevent Robber from entering $V_C$ forever.
So Robber must start at the vertex $r$. Then observe, that $c$ cops must occupy the positions $S$
and $m$ cops must occupy the vertices $T$. If any cop does not start either on $T$ or $S$,
Robber wins in the next move. The cops on $T$ remain there harmless to the end of the game.
The cops cannot move until Robber decides to leave the vertex $r$.
After that, the vertices in $V'\setminus V$ no longer affect the game,
thus $\G'$ exactly imitates $\G$.

The construction of $\dir G'$ from $\dir G$ can be realised by an iteration over the set
of vertices and edges of $\dir G$, which means that the construction can be done in log-space.
It is easy to see, that the number of vertices and edges we add is linear in the number of vertices
of $\dir G$, proving the last statement.
\qed\end{proof}

Here we give the precise definition of log-space reductions via length order.
Let $\Sigma=\{0,1\}$ denote the alphabet and let $\Sigma^+$ denote all nonempty
and finite words over the alphabet $\Sigma$.
The function $f:\Sigma^+\to\Sigma^+$ is \emph{logspace-computable} if there
is a deterministic Turing machine with a separate two-way read-only input tape,
a read/write work tape, and a one-way output tape such that, when started with any
word $w\in\Sigma^+$ on the input tape, the machine eventually halts
with $f(w)$ on the output tape while having visited at most $\log|w|$ cells on the work tape.
Let $\ell:\N\to\R$. The function $f$ is \emph{length $\ell(n)$ bounded}
if $|f(w)|\le \ell(|w|)$ for all $w\in\Sigma^+$.

Let $A,B\subseteq\Sigma^+$.
\emph{$A$ transforms to $B$ within logspace via $f$} (denoted by $A\logred B$ via $f$)
if $f:\Sigma^+\to\Sigma^+$ is logspace-computable function such that
$w\in A \Leftrightarrow f(w)\in B$ for all $w\in\Sigma^+$.

Let $B$ be a language and let $\mathcal{L}$ be a class of languages.
Then $\mathcal{L}\logred B$ if $A\logred B$ for all $A\in\mathcal{L}$.
Furthermore, $\mathcal{L}\logred B$ \emph{via length order $\ell(n)$, $\ell:\N\to\R$},
provided that for each $A\in\mathcal{L}$ there is a function $f$ and constant $b\in\N$
such that $A\logred B$ via $f$ and $f$ is length $b\ell(n)$ bounded.
The language $B$ is \emph{$\mathcal{L}$-complete under log-space reductions}
if both $B\in\mathcal{L}$ and $\mathcal{L}\logred B$.

\begin{proof}[of Theorem \ref{t:main}]
By Lemma~\ref{l:etime}, $\Guard\in\E$.
Consider a problem $L \in \E$.
By Theorem~\ref{t:stocha}, $L$ can be reduced to the formula satisfying game $\F$ in log-space
via length-order $n \log n$.
By Theorem~\ref{t:prescribedgame} there exists an equivalent (in the terms of the
game outcome) guarding game $\G=(\dir G,V_C,c, S,r)$ with prescribed starting positions,
together with other properties as stated by Theorem~\ref{t:prescribedgame}.
By Lemma~\ref{l:force} applied on the game $\G$
there is an equivalent guarding game $\G'$, $\G\subseteq \G'$, without prescribed starting positions.
Thus, we reduced $L$ to \Guard.
By Lemma~\ref{l:logspace} and Lemma~\ref{l:force}, the whole reduction can be done in log-space,
via length order $n (\log n)^2$.
Therefore, Theorem~\ref{t:main} is proved.
\qed\end{proof}

\begin{proof}[of Corollary \ref{c:maintime}]
By the time hierarchy theorem, if $1 < a < b$, there is a problem $X$ in
$\DTIME(b^n) \setminus \DTIME(a^n)$.
Thus, $X$ is in $\E$ and we can use the reduction from Theorem~\ref{t:main}.
For an instance $I$ of $X$ of size $n$, we can
construct in polynomial time an instance $I'$ of \Guard\ of size $m = O(n (\log n)^2)$ so that
deciding $I'$ solves the instance~$I$.
By choice of $n$, this cannot be done in time $a^n$ for all $n$.
It remains to check that $n = \Omega(m/(\log m)^2)$.
\qed\end{proof}

\section*{Further questions}

As we have already mentioned, the relation of the classes \PSPACE\ and \E\ is
unclear as we only know that $\PSPACE \ne \E$ and the current state of the art
is missing some deeper understanding of the relation. Therefore, the conjecture
of Fomin et al.\ whether \Guard\ is \PSPACE-complete still remains open.
However, we believe that the conjecture is not true.

For a guarding game $\G=(G,V_C,c)$, what happens if we restrict the size of strongly connected components of $G$?
If the sizes are restricted by 1, we get DAG, for which the decision problem is \PSPACE-complete.
For unrestricted sizes we have shown that $\G$ is \E-complete.
Is there some threshold for $\G$ to become \E-complete from being \PSPACE-complete?
This may give us some insight into the original conjecture.
Another interesting question is what happens if we bound the degrees of~$G$ by a constant.
%Furthermore, what happens if we restrict the degrees of vertices of $G$ or the treewidth of $G$?
We are also working on forcing the starting position in the guarding game on undirected graphs
in a way similar to Theorem~\ref{t:main}.

\section*{Acknowledgements}

We would like to thank Ruda Stola\v r for drawing initial versions of some
pictures and for useful discussion.
We thank Peter Golovach for giving a nice talk about the problem,
which inspired us to work on it.
We would also like to thank Jarik Ne\v set\v ril for suggesting some of the previous open questions
and to Honza Kratochv\'\i l for fruitful discussion of the paper structure.

\end{document}